\documentclass[journal]{IEEEtran}
\usepackage{amsmath}
\usepackage{amssymb}
\usepackage{amsfonts}
\usepackage{graphicx}
\usepackage{epsfig}
\usepackage{subfigure}
\usepackage{psfrag}
\usepackage{color}
\usepackage{cite}

\begin{document}

\title{Wireless Information and Power Transfer in Multiuser OFDM Systems}

\author{Xun Zhou, Rui Zhang,~\IEEEmembership{Member,~IEEE}, and Chin Keong Ho,~\IEEEmembership{Member,~IEEE}
\thanks{This paper has been presented in part at IEEE Global Communications
Conference (Globecom), December 9-13, 2013, Atlanta, USA.}
\thanks{X. Zhou is with the Department of
Electrical and Computer Engineering, National University of Singapore
(e-mail: xunzhou@nus.edu.sg).}
\thanks{R. Zhang is with the Department
of Electrical and Computer Engineering,
National University of Singapore (e-mail: elezhang@nus.edu.sg).
He is also with the Institute for Infocomm Research, A*STAR, Singapore.}
\thanks{C. K. Ho is with the Institute for Infocomm Research, A*STAR,
Singapore (e-mail: hock@i2r.a-star.edu.sg).}}

\maketitle

\begin{abstract}
In this paper, we study the optimal design for simultaneous wireless
information and power transfer (SWIPT) in downlink multiuser
orthogonal frequency division multiplexing (OFDM) systems, where
the users harvest energy and decode information using the same signals
received from a fixed access point (AP). For information transmission,
we consider two types of multiple access schemes, namely,
time division multiple access (TDMA) and orthogonal frequency
division multiple access (OFDMA). At the receiver side,
due to the practical limitation
that circuits for harvesting energy from radio signals are not yet able
to decode the carried information directly, each user applies either
time switching (TS) or power splitting (PS) to coordinate the
energy harvesting (EH) and information decoding (ID) processes.
For the TDMA-based information transmission, we employ TS at
the receivers; for the OFDMA-based information transmission, we employ
PS at the receivers. Under the above two scenarios,
we address the problem of maximizing the weighted sum-rate
over all users by varying the time/frequency power allocation
and either TS or PS ratio, subject to a minimum harvested energy
constraint on each user as well as a peak and/or total transmission power
constraint. For the TS scheme, by an appropriate variable transformation
the problem is reformulated as a convex problem, for which the optimal
power allocation and TS ratio are obtained by the Lagrange duality method.
For the PS scheme, we propose an iterative algorithm to optimize the power
allocation, subcarrier (SC) allocation and the PS ratio for each user.
The performances of the two schemes are compared numerically as well as analytically
for the special case of single-user setup. It is revealed that the peak power
constraint imposed on each OFDM SC as well as the number of users
in the system play key roles in the rate-energy performance comparison by the
two proposed schemes.
\end{abstract}

\begin{keywords}
Simultaneous wireless information and power transfer (SWIPT), energy harvesting,
wireless power,
orthogonal frequency division multiplexing (OFDM), orthogonal frequency
division multiple access (OFDMA), time division multiple access (TDMA),
time switching, power splitting.
\end{keywords}

\IEEEpeerreviewmaketitle

\setlength{\baselineskip}{1.0\baselineskip}
\newtheorem{definition}{\underline{Definition}}[section]
\newtheorem{fact}{Fact}
\newtheorem{assumption}{Assumption}
\newtheorem{theorem}{\underline{Theorem}}[section]
\newtheorem{lemma}{\underline{Lemma}}[section]
\newtheorem{corollary}{\underline{Corollary}}[section]
\newtheorem{proposition}{\underline{Proposition}}[section]
\newtheorem{example}{\underline{Example}}[section]
\newtheorem{remark}{\underline{Remark}}[section]
\newtheorem{algorithm}{\underline{Algorithm}}[section]

\newcommand{\mv}[1]{\mbox{\boldmath{$ #1 $}}}
\newcommand{\Peak}{P_{\rm peak}}
\newcommand{\N}{\Gamma\sigma^2}
\newcommand{\Emin}{\overline{E}_k}
\newcommand{\Heh}{h_{k,n}^{\rm EH}}
\newcommand{\Hid}{h_{k,n}^{\rm ID}}

\section{Introduction}
Recently, simultaneous wireless information and power transfer (SWIPT) becomes appealing by essentially providing a perpetual energy source for the wireless networks \cite{Zhang}.
Moreover, the SWIPT system offers great convenience to mobile users, since it realizes both useful utilizations of radio signals to transfer energy as well as information. Therefore, SWIPT has drawn an upsurge of research interests \cite{Zhang,Varshney,Xun,Xiang,Ju,Xu,Liu-a,Liu-b,Timotheou,Park,Nasir,Huang-cellular,Lee,Sahai,Huang-OFDM,Derrick-OFDM}.
Varshney first proposed the idea of transmitting information and energy
simultaneously in \cite{Varshney} assuming
that the receiver is able to decode information
and harvest energy simultaneously from the same received signal. However, this assumption may not hold in
practice, as circuits for harvesting energy from radio signals are not yet able to decode the
carried information directly. Two practical schemes for SWIPT, namely,
time switching (TS) and power splitting (PS), are proposed in \cite{Zhang,Xun}. With TS applied at the receiver,
the received signal is either processed by an energy receiver for energy harvesting (EH) or processed by
an information receiver for information decoding (ID). With PS applied at the receiver, the received signal is split
into two signal streams with a fixed power ratio by a power splitter, with one stream to the energy receiver and
the other one to the information receiver.
SWIPT for multi-antenna systems has been considered in \cite{Zhang,Xiang,Xu,Ju}.
In particular, \cite{Zhang} studied the performance limits of a three-node multiple-input multiple-output (MIMO)
broadcasting system, where one receiver harvests energy and another receiver decodes information from the signals sent
by a common transmitter. \cite{Xiang} extended the work in \cite{Zhang} by
considering imperfect channel state information (CSI) at the transmitter for a multiple-input single-output (MISO) system.
A MISO SWIPT system without CSI at the transmitter was considered in \cite{Ju}, where a new scheme that employs random beamforming for opportunistic EH was proposed.
\cite{Xu} studied a MISO broadcasting system that exploits near-far channel conditions,
where ``near'' users are scheduled for EH, while ``far'' users are scheduled for ID.
SWIPT that exploits flat-fading channel variations was studied in \cite{Liu-a,Liu-b},
where the receiver performs dynamic time switching (DTS) \cite{Liu-a}
or dynamic power splitting (DPS) \cite{Liu-b} to coordinate between EH and ID.
SWIPT in multi-antenna interference channels was considered in \cite{Timotheou,Park} using PS and TS, respectively.
SWIPT with energy/information relaying has been studied in \cite{Nasir}, where an energy-constrained relay
harvests energy from the received signal and uses that harvested energy to forward the source information to
the destination.
Two relaying protocols, i.e., the TS-based relaying (TSR) protocol and the PS-based relaying (PSR) protocol, are proposed
in \cite{Nasir}. In the TSR protocol, the relay spends a portion of time for EH and the remaining time for information processing. In the PSR protocol, the relay spends a portion of the received power for EH and the remaining power
for information processing.
Networks that involve wireless power transfer were studied in \cite{Huang-cellular,Lee}.
In \cite{Huang-cellular}, the authors studied a hybrid network which
overlays an uplink cellular network with randomly deployed
power beacons that charge mobiles wirelessly. Under an outage
constraint on the data links, the tradeoffs among the network
parameters were derived. In \cite{Lee}, the authors investigated a
cognitive radio network powered by opportunistic wireless
energy harvesting, where mobiles from the secondary network
either harvest energy from nearby transmitters in a primary
network, or transmit information if the primary transmitters are
far away. Under an outage constraint for coexisting networks,
the throughput of the secondary network was maximized.

Orthogonal frequency division multiplexing (OFDM) is a well established technology for high-rate wireless communications, and
has been adopted in various standards, e.g., IEEE 802.11n and 3GPP-Long Term Evolution (LTE).
However, the performance may be limited by the availability of energy in the devices for some practical application scenarios.
It thus motivates our investigation of SWIPT in OFDM-based systems.
SWIPT over a single-user OFDM channel has been studied
in \cite{Sahai} assuming that the receiver is able to decode information
and harvest energy simultaneously from the same received signal.
It is shown in \cite{Sahai} that a tradeoff exists between the achievable
rate and the transferred power by power allocation in the frequency bands:
for sufficiently small transferred power, the optimal power allocation is
given by the so-called ``waterfilling (WF)'' allocation to maximize the
information transmission rate, whereas
as the transferred power increases, more power needs to be allocated to the
channels with larger channel gain and finally approaches the strategy with
all power allocated to the channel with largest channel gain.
However, due to the practical limitation that circuits for harvesting energy from radio signals are not yet able to decode the
carried information directly, the result in \cite{Sahai} actually provides only an upper bound for the rate-energy tradeoff in a single-user OFDM system.
Power control for SWIPT in a multiuser multi-antenna OFDM system was considered in \cite{Huang-OFDM}, where
the information decoder and energy harvester are attached to two separate antennas.
In \cite{Huang-OFDM}, each user
only harvests the energy carried by the subcarriers that are allocated to that user for ID, which is inefficient in energy utilization, since the energy carried by the subcarriers allocated to other users for ID can be potentially harvested.
Moreover, \cite{Huang-OFDM} focuses on power control by assuming a predefined subcarrier allocation. In this paper, we jointly optimize the power allocation strategy as well as the subcarrier allocation strategy.
\cite{Derrick-OFDM} considered SWIPT in a multiuser single-antenna OFDM system, where PS is applied at each receiver to coordinate between EH and ID. In \cite{Derrick-OFDM}, it is assumed that the splitting ratio can be different for different subcarriers.
However, in practical circuits, (analog) power splitting is performed before (digital) OFDM demodulation.
Thus, for an OFDM-based SWIPT system, {\em all} subcarriers would have to be power split with the same power ratio at each receiver even though only a subset of the subcarriers contain information for the receiver. In contrast, for the case of a single-carrier system, a receiver simply harvests energy from all signals that do not contain information for this receiver.

As an extension of our previous work in \cite{Xun} for a single-user narrowband SWIPT system, in this paper, we study a multiuser OFDM-based SWIPT system (see Fig. \ref{fig:system}), where a fixed access point (AP) with constant power supply broadcasts signals to a set of distributed users.
Unlike the conventional wireless network where all the users contain only information receiver and draw energy from their own power supplies, in our model,
it is assumed that each user contains an additional energy receiver to harvest energy from the received signals from the AP.
For the information transmission, two conventional multiple access schemes are considered, namely, time division multiple access (TDMA) and orthogonal frequency division multiple access (OFDMA). For the TDMA-based information transmission, since
the users are scheduled in nonoverlapping time slots to decode information, each user should apply TS such that the information receiver is used during the time slot when information is scheduled for that user, while the energy receiver is used in all other time slots. For the OFDMA-based information transmission, we assume that PS is applied at each receiver. As mentioned in the previous paragraph, we assume that all subcarriers share the same power splitting ratio at each receiver.
Under the TDMA scenario,
we address the problem of maximizing the weighted sum-rate over all users by varying the power allocation in time and frequency and the TS ratios, subject to a minimum harvested energy constraint on each user and a peak and/or total transmission power constraint.
By an appropriate variable transformation the problem is reformulated as a convex problem, for which the optimal power allocation and TS ratios are obtained by the Lagrange duality method.
For the OFDMA scenario, we address the same problem by varying the power allocation in frequency, the subcarrier allocation to users and the PS ratios.
In this case, we propose an efficient algorithm to iteratively optimize the power and subcarrier allocation, and the PS ratios at receivers until the convergence is reached.
Furthermore, we compare the rate-energy performances by the two proposed schemes, both numerically by simulations and analytically
for the special case of single-user system setup.
It is revealed that the peak power
constraint imposed on each OFDM subcarrier as well as the number of users
in the system play key roles in the rate-energy performance comparison by the
two proposed schemes.

The rest of this paper is organized as follows. Section \ref{sec:system model} presents the system model and problem formulations. Section \ref{sec:solution-single} studies the special case of a single-user OFDM-based SWIPT system.
Section \ref{sec:solution} derives the resource allocation solutions for the two proposed schemes in the multiuser OFDM-based SWIPT system. Finally, Section \ref{sec:conclusion} concludes the paper.

\section{System Model and Problem Formulation}\label{sec:system model}
\begin{figure}
\begin{center}
\scalebox{0.5}{\includegraphics{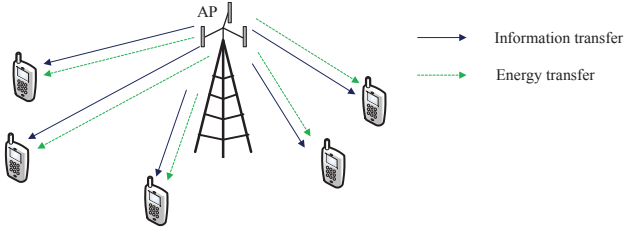}}
\end{center}
\caption{A multiuser downlink SWIPT system.}
\label{fig:system}
\end{figure}

As shown in Fig. \ref{fig:system}, we consider a downlink OFDM-based system with
one transmitter and $K$ users.
The transmitter and all users are each equipped with one antenna.
The total bandwidth of the system is equally divided into $N$ subcarriers (SCs).
The SC set is denoted by $\mathcal{N}=\{1,\ldots,N\}$.
The power allocated to SC $n$ is denoted by $p_n,n=1,\ldots,N$.
Assume that the total transmission power is at most $P$.
The maximum power allocated to each SC is denoted by $\Peak$, i.e.,
$0\leq p_n\leq \Peak,\forall n\in\mathcal{N}$, where $\Peak\geq P/N$.
The channel power gain of SC $n$ as seen by the user $k$ is denoted
by $h_{k,n},k=1,\ldots,K,n=1,\ldots,N$.
We consider a slow-fading environment, where all the channels are assumed to be constant within
the transmission scheduling period of our interest.
For simplicity, we assume the total transmission time to be one.
Moreover, it is assumed that the channel
gains on all the SCs for all the users are known at the transmitter.
At the receiver side, each user performs EH in addition to ID.
It is assumed that the minimum harvested energy during the unit transmission time
is $\Emin>0$ for user $k,k=1,\ldots,K$.

\subsection{TDMA with Time Switching}
We first consider the case of TDMA-based information transmission with TS applied at each receiver.
It is worth noting that for a single-user SWIPT system with TS applied at the receiver, the transmission time needs to be divided into two time slots to coordinate the EH and ID processes at the receiver.
Thus, in the SWIPT system with $K$ users, we consider $K+1$ time slots without loss of generality, where the additional time slot, which we called the {\em power} slot, may be allocated for all users to perform EH only. In contrast, in conventional TDMA systems without EH, the power slot is not required.
We assume that slot $k,k=1,\ldots,K$ is assigned to user $k$ for transmitting information,
while slot $K+1$ is the power slot.
With total time duration of $K+1$ slots to be at most one,
the (normalized) time duration for slot $k,k=1,\ldots,K+1$ is variable and
denoted by the TS ratio $\alpha_k$, with $0\leq\alpha_k\leq 1$ and $\sum\limits_{k=1}^{K+1}\alpha_k\leq 1$.
In addition, the power $p_n$ allocated to SC $n$ at slot $k$ is specified as $p_{k,n}$,
where $0\leq p_{k,n}\leq \Peak,k=1,\ldots,K+1,n=1,\ldots,N$.
The average transmit power constraint is thus given by
\begin{equation}\label{}
    \sum\limits_{k=1}^{K+1}\alpha_k\sum\limits_{n=1}^{N}p_{k,n} \leq P.
\end{equation}

Consider user $k,k=1,\ldots,K.$ At the receiver side, user $k$ decodes its intended information at slot $k$ when its information is sent and harvests energy during all the other slots $i\neq k$.
The receiver noise at each user is assumed to be independent over SCs and is modelled as a circularly symmetric complex Gaussian (CSCG) random variable with zero mean and variance $\sigma^2$ at all SCs.
Moreover, the gap for the achievable rate from the channel capacity due to a practical modulation and coding scheme (MCS)
is denoted by $\Gamma\geq1$.
The achievable rate in bps/Hz for the information receiver of user $k$ is thus given by
\begin{equation}
    R_k = \frac{\alpha_k}{N}\sum\limits_{n=1}^{N} \log_2 \left(1+\frac{h_{k,n}p_{k,n}}{\N}\right).
\end{equation}
Assuming that the conversion efficiency of the energy harvesting process at each receiver
is denoted by $0<\zeta<1$, the harvested energy in joule at the energy receiver of user $k$ is thus given by
\begin{equation}\label{eq:Ek-TS}
    E_k = \zeta\sum\limits_{i\neq k}^{K+1}\alpha_i\sum\limits_{n=1}^{N} h_{k,n} p_{i,n}.
\end{equation}

\begin{figure}
\begin{center}
\scalebox{0.55}{\includegraphics{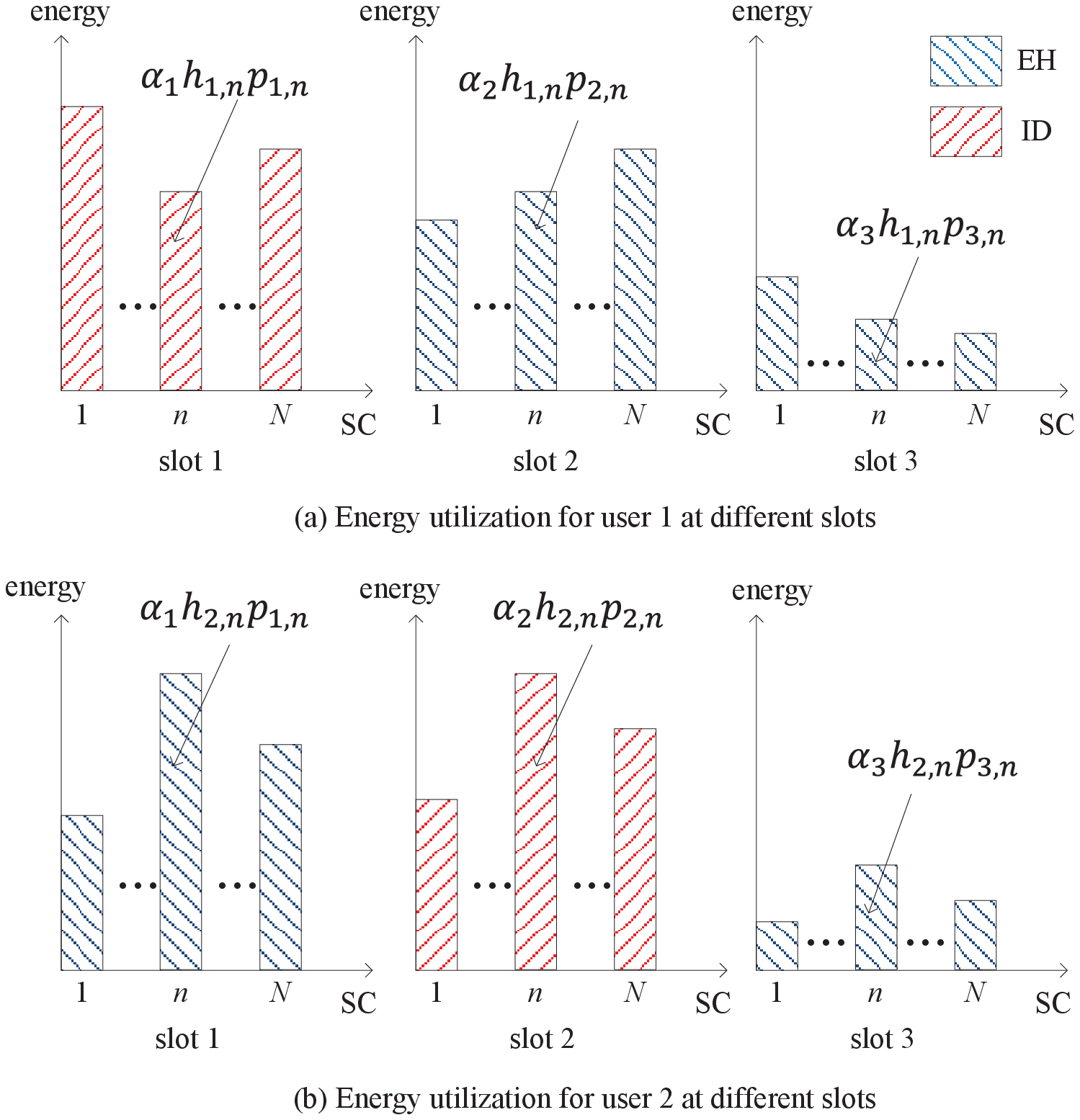}}
\end{center}
\caption{Energy utilization at receivers for a two-user OFDM-based SWIPT system: TDMA-based information transmission with TS applied at each receiver. In sub-figure (a) for user 1, the received energy on all SCs during slot 1 is utilized for ID; while the received energy on all SCs during slot 2 and slot 3 is utilized for EH. In sub-figure (b) for user 2, the received energy on all SCs during slot 2 is utilized for ID; while the received energy on all SCs during slot 1 and slot 3 is utilized for EH.}
\label{fig:PA-TS}
\end{figure}

An example of the energy utilization at receivers for the TS case in a two-user OFDM-based SWIPT system is illustrated in Fig. \ref{fig:PA-TS}. As shown in Fig. \ref{fig:PA-TS}(a) for user 1, the received energy on all SCs during slot 1 is utilized for ID; while the received energy on all SCs during slot 2 and slot 3 is utilized for EH. In Fig. \ref{fig:PA-TS}(b) for user 2, the received energy on all SCs during slot 2 is utilized for ID; while the received energy on all SCs during slot 1 and slot 3 is utilized for EH.

Our objective is to maximize the weighted sum-rate of all users by varying the transmission power in the time and frequency domains jointly with TS ratios, subject to EH constraints and the transmission
power constraints. Thus, the following optimization problem is formulated.

\begin{small}
\begin{align}
\mathrm{(P-TS)}:\hspace{10pt}\nonumber\\
\hspace{-5pt}~\mathop{\mathtt{max.}}_{\{p_{k,n}\},\{\alpha_k\}} & ~~ \frac{1}{N}\sum\limits_{k=1}^{K}\sum\limits_{n=1}^{N} w_k\alpha_k \log_2 \left(1+\frac{h_{k,n}p_{k,n}}{\N}\right) \nonumber \\
\mathtt{s.t.} & ~~ \zeta\sum\limits_{i\neq k}^{K+1}\alpha_i\sum\limits_{n=1}^{N} h_{k,n} p_{i,n} \geq \Emin,
    ~ ~ k=1,\ldots,K, \nonumber\\
              & ~~ \sum\limits_{k=1}^{K+1}\alpha_k\sum\limits_{n=1}^{N}p_{k,n}\leq P, \nonumber\\
              & ~~ 0\leq p_{k,n}\leq \Peak,  k=1,\ldots,K+1,\forall n, \nonumber \\
              & ~~ \sum\limits_{k=1}^{K+1} \alpha_k \leq 1, ~~~~ 0\leq\alpha_k\leq 1,k=1,\ldots,K+1,\nonumber
\end{align}
\end{small}%
where $w_k\geq 0$ denotes the non-negative rate weight assigned to user $k$.

Problem (P-TS) is feasible when all the constraints in Problem (P-TS) can be satisfied by some $\left\{\{p_{k,n}\},\{\alpha_k\}\right\}$. From (\ref{eq:Ek-TS}), the harvested energy at all users is maximized when $\alpha_{K+1}=1$, while $\alpha_k=0,p_{k,n}=0$ for $k=1,\ldots,K,n=1,\ldots,N$, i.e., all users harvest energy during the entire transmission time. Therefore, Problem (P-TS) is feasible if and only if the following linear programming (LP) is feasible.
\begin{align}
~\mathop{\mathtt{max.}}_{\{p_{K+1,n}\}} & ~~ 0 \nonumber \\
\mathtt{s.t.} & ~~ \zeta\sum\limits_{n=1}^{N} h_{k,n} p_{K+1,n} \geq \Emin,
    ~ ~ k=1,\ldots,K, \nonumber\\
              & ~~ \sum\limits_{n=1}^{N}p_{K+1,n}\leq P, \nonumber\\
              & ~~ 0\leq p_{K+1,n}\leq \Peak, ~~ n=1,\ldots,N.   \label{prob:feasibility}
\end{align}
It is easy to check the feasibility for the above LP. We thus assume Problem (P-TS) is feasible subsequently.

Problem (P-TS) is non-convex in its current form.
We will solve this problem in Section \ref{sec:solution-TS}.

\subsection{OFDMA with Power Splitting}

Next, we consider the case of OFDMA-based information transmission with PS applied at each receiver.
As is standard in OFDMA transmissions, each SC is allocated to at most one user in each slot, i.e., no SC sharing is allowed. We define a SC allocation function $\Pi(n)\in\{1,\ldots,K\}$, i.e., the SC $n$ is allocated to user $\Pi(n)$.
The total transmission power constraint is given by
\begin{equation}\label{}
  \sum\limits_{n=1}^{N} p_n \leq P.
\end{equation}
At the receiver, the received signal
at user $k$ is processed by a power splitter,
where a ratio $\rho_k$ of power is split to its energy receiver and the
remaining ratio $1-\rho_k$ is split to its information receiver, with $0\leq\rho_k\leq 1,\forall k$.
The achievable rate in bps/Hz at SC $n$ assigned to user $\Pi(n)$ is thus
\begin{equation}
    R_n = \log_2 \left(1+\frac{(1-\rho_{\Pi(n)})h_{\Pi(n),n}p_n}{\N}\right), ~~ n=1,\ldots,N.
\end{equation}
With energy conversion efficiency $\zeta$, the harvested energy in joule at the energy receiver of user $k$ is thus given by
\begin{equation}\label{eq:Ek-PS}
    E_k = \rho_k\zeta\sum\limits_{n=1}^{N} h_{k,n} p_n, ~ ~ k=1,\ldots,K.
\end{equation}

\begin{figure}
\begin{center}
\scalebox{0.55}{\includegraphics{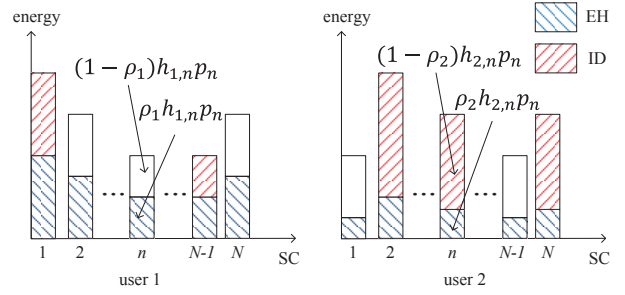}}
\end{center}
\caption{Energy utilization at receivers for a two-user OFDM-based SWIPT system: OFDMA-based information transmission with PS applied at each receiver. The received signals at all SCs share the same splitting ratio $\rho_k$ at each user $k,k=1,2$.}
\label{fig:PA-PS}
\end{figure}

An example of the energy utilization at receivers for the PS case in a two-user OFDM-based SWIPT system is illustrated in Fig. \ref{fig:PA-PS}. As shown in Fig. \ref{fig:PA-PS}, the received signals at all SCs share the same splitting ratio $\rho_k$ at each user $k,k=1,2$. It is worth noting that only $\rho_1$ of the power at each of the SCs allocated to user 2 for ID is harvested by user 1, the remaining $1-\rho_1$ of power at those SCs is neither utilized for EH nor ID at user 1, similarly as for user 2 with PS ratio $\rho_2$.

With the objective of maximizing the weighted sum-rate of all users by varying the transmission power in the frequency domain, the SC allocation, jointly with the PS ratios at receivers, subject to a given set of EH constraints and the transmission power constraints, the following optimization problem is formulated.

\begin{small}
\vspace{-0.1in}
\begin{align}
\mathrm{(P-PS)}:\hspace{18pt}\nonumber\\
\mathop{\mathtt{max.}}_{\{p_n\},\{\Pi(n)\},\{\rho_k\}} & ~~ \frac{1}{N}\sum\limits_{n=1}^{N}w_{\Pi(n)}\log_2 \left(1+\frac{(1-\rho_{\Pi(n)})h_{\Pi(n),n}p_n}{\N}\right) \nonumber \\
\mathtt{s.t.} & ~~ \rho_k\zeta\sum\limits_{n=1}^{N} h_{k,n} p_n \geq \Emin, ~ ~ \forall k \nonumber\\
              & ~~ \sum\limits_{n=1}^{N} p_n \leq P, ~~ 0\leq p_n\leq \Peak,\forall n \nonumber\\
              & ~~ 0\leq\rho_k\leq 1, ~~ \forall k. \nonumber
\end{align}\vspace{-0.1in}
\end{small}

From (\ref{eq:Ek-PS}), the harvested energy at all users is maximized when $\rho_k=1,k=1,\ldots,K$, i.e., all power is split to the energy receiver at each user. Therefore, Problem (P-PS) is feasible if and only if Problem (P-PS) with $\rho_k=1,k=1,\ldots,K$ is feasible. It is worth noting that Problem (P-PS) and Problem (P-TS) are subject to the same feasibility conditions as given by Problem (\ref{prob:feasibility}).

It can be verified that Problem (P-PS) is non-convex in its current form. We will solve this problem in Section \ref{sec:solution-PS}.

\subsection{Performance Upper Bound}
An upper bound for the optimization problems (P-TS) and (P-PS) can be obtained by assuming that each receiver is able to decode the information in its received signal and at the same time harvest the received energy without any implementation loss \cite{Sahai}.
We thus consider the following optimization problem.

\begin{small}
\vspace{-0.1in}
\begin{align}
\mathrm{(P-UB)}:
~\mathop{\mathtt{max.}}_{\{p_n\},\{\Pi(n)\}} & ~~ \frac{1}{N}\sum\limits_{n=1}^{N}w_{\Pi(n)}\log_2 \left(1+\frac{h_{\Pi(n),n}p_n}{\N}\right) \nonumber \\
\mathtt{s.t.} & ~~ \zeta\sum\limits_{n=1}^{N} h_{k,n} p_n \geq \Emin, ~ ~ \forall k \nonumber\\
              & ~~ \sum\limits_{n=1}^{N} p_n \leq P, ~~ 0\leq p_n\leq \Peak,\forall n. \nonumber
\end{align}
\end{small}%
Note that Problem (P-UB), as well as Problem (P-TS) and Problem (P-PS) are subject to the same feasibility conditions as given by Problem (\ref{prob:feasibility}).
Also note that any infeasible Problem (P-UB) can be modified to become a feasible one by increasing the transmission power $P$ or by decreasing the minimum required harvested energy $\Emin$ at some users $k$.
In the sequel, we assume that all the three problems are feasible, thus optimal solutions exist.

The solution for Problem (P-UB) is obtained in Section \ref{sec:solution-PS} (see Remark \ref{remark:P-UB}).

\section{Resource Allocation in a Single-user System}\label{sec:solution-single}
To obtain tractable analytical results, in this section, we consider the special case that $K=1$, i.e., a single-user OFDM-based SWIPT system.
For brevity, $h_{1,n}$, $\overline{E}_1$, and $\rho_1$ are replaced with $h_n$, $\overline{E}$, and $\rho$, respectively.
Without loss of generality, we assume that $h_1\geq h_2\geq \ldots\geq h_N$
and $w_1=1$.
With $K=1$, Problem (P-TS) and Problem (P-PS) are then simplified respectively as follows
\begin{align}
\mathop{\mathtt{max.}}_{\{p_{1,n}\},\{p_{2,n}\},\alpha_1,\alpha_2} & ~~ \frac{\alpha_1}{N}\sum\limits_{n=1}^{N} \log_2 \left(1+\frac{h_{n}p_{1,n}}{\N}\right) \nonumber \\
\mathtt{s.t.} & ~~ \zeta\alpha_2\sum\limits_{n=1}^{N} h_{n} p_{2,n} \geq \overline{E},		\nonumber\\
              & ~~ \alpha_1\sum\limits_{n=1}^{N}p_{1,n}+\alpha_2\sum\limits_{n=1}^{N}p_{2,n}\leq P, \nonumber\\
              & ~~ 0\leq p_{i,n}\leq \Peak, ~~ \forall n, i=1,2, \nonumber \\
              & ~~ \alpha_1+\alpha_2 \leq 1, ~~ 0\leq\alpha_i\leq 1, i=1,2.			\label{prob:TS-single}
\end{align}
\begin{align}
~\mathop{\mathtt{max.}}_{\{p_n\},\rho} & ~~ \frac{1}{N}\sum\limits_{n=1}^{N}\log_2 \left(1+\frac{(1-\rho)h_{n}p_n}{\N}\right) \nonumber \\
\mathtt{s.t.} & ~~ \rho\zeta\sum\limits_{n=1}^{N} h_{n} p_n \geq \overline{E},  \nonumber\\
              & ~~ \sum\limits_{n=1}^{N} p_n \leq P, \nonumber\\
              & ~~ 0\leq p_n\leq \Peak, ~~ \forall n, \nonumber\\
              & ~~ 0\leq\rho\leq 1.           \label{prob:PS-single}
\end{align}

To obtain useful insight, we first look at the two extreme cases, i.e., $\Peak\rightarrow\infty$ and $\Peak=P/N$.
We shall see that the peak power constraint plays an important role in the performance comparison between the TS and PS schemes.
Note that $\Peak\rightarrow\infty$ implies the case of no peak power constraint on each SC; and $\Peak=P/N$ implies the case of only peak power constraint on each SC, since the total power constraint is always satisfied and thus becomes redundant. Given $P$ and $\Peak$, the maximum rates achieved by the TS scheme and the PS scheme are denoted by $R_{\rm TS}(P,\Peak)$ and $R_{\rm PS}(P,\Peak)$, respectively.
For the case of $\Peak\rightarrow\infty$, we have the following proposition for the TS scheme.
We recall that $\alpha_2=1-\alpha_1$ is the TS ratio for the power slot.
\begin{proposition}\label{proposition-1}
Assuming $\overline{E}>0$, in the case of a single-user OFDM-based SWIPT system with $\Peak\rightarrow\infty$, the maximum rate by the TS scheme, i.e., $R_{\rm TS}(P,\infty)$, is achieved by
$\alpha_1\rightarrow 1$ and $\alpha_2\rightarrow 0$.
\end{proposition}
\begin{proof}
Clearly, we have $\alpha_2>0$; otherwise, no energy is harvested, which violates
the EH constraint $\overline{E}>0$. Thus, $\alpha_1<1$.
To maximize the objective function subject to the EH constraint, it can be easily
shown that the optimal $\alpha_2$ and $p_{2,n}$ should satisfy $\zeta \alpha_2 h_1 p_{2,1}=\overline{E}$ and $p_{2,n}=0,n=2,\ldots,N$.
It follows that the minimum transmission energy consumed to achieve the harvested energy $\overline{E}$
is given by $\overline{E}/(\zeta h_1)$, i.e., $\alpha_2\sum\limits_{n=1}^{N}p_{2,n}\geq\overline{E}/(\zeta h_1)$.
Thus, in Problem (\ref{prob:TS-single}), the achievable rate $R_{\rm TS}(P,\infty)$ is given by maximizing $\frac{\alpha_1}{N}\sum\limits_{n=1}^{N}\log_2 \left(1+\frac{h_{n}p_{1,n}}{\N}\right)$ subject to $\alpha_1\sum\limits_{n=1}^{N}p_{1,n}\leq P-\overline{E}/(\zeta h_1)$ and $0\leq\alpha_1<1$.
Let $q_{1,n}=\alpha_1 p_{1,n},\forall n$, the above problem is then equivalent to maximizing
$\frac{\alpha_1}{N}\sum\limits_{n=1}^{N}\log_2 \left(1+\frac{h_{n}q_{1,n}}{\N \alpha_1}\right)$ subject to
$\sum\limits_{n=1}^{N}q_{1,n}\leq P-\overline{E}/(\zeta h_1)$ and $0\leq\alpha_1<1$.
For given $\{q_{1,n}\}$, the objective function
is an increasing function of $\alpha_1$; thus, $R_{\rm TS}(P,\infty)$ is maximized when $\alpha_1\rightarrow1$.
It follows that $\alpha_2\rightarrow 0$, which completes the proof.
\end{proof}

\begin{remark}\label{remark:Peak-infty}
By Proposition \ref{proposition-1}, to achieve $R_{\rm TS}(P,\infty)$ with $\overline{E}>0$, the portion of
transmission time $\alpha_2$ allocated to EH in each transmission block should asymptotically go to zero.
For example, let $m$ denote the number of transmitted symbols in each block, by allocating
$O(\log m)$ symbols for EH in each block and the remaining symbols for ID results in $\alpha=\log m/m\rightarrow0$ as $m\rightarrow\infty$, which satisfies the optimality condition provided in Proposition \ref{proposition-1}.
It is worth noting that $R_{\rm TS}(P,\infty)$ is achieved under the assumption that the transmitter and receiver are able to operate in the regime of infinite power in the EH time slot due to $\alpha_2\rightarrow 0$.
For a finite $\Peak$, a nonzero time ratio should be scheduled to the power slot to collect sufficient energy to satisfy the EH constraint.
\end{remark}

Moreover, we have the following proposition showing that the PS scheme performs no better than the TS scheme for the case of $\Peak\rightarrow\infty$.
\begin{proposition}\label{proposition-2}
In the case of a single-user OFDM-based SWIPT system with $\Peak\rightarrow\infty$, the maximum rate achieved by the PS scheme is no larger than that achieved by the TS scheme, i.e., $R_{\rm PS}(P,\infty)\leq R_{\rm TS}(P,\infty)$.
\end{proposition}
\begin{proof}
For the PS scheme, from the EH constraint $\rho\zeta\sum\limits_{n=1}^N h_n p_n\geq\overline{E}$, it follows
that $\rho\geq \overline{E}/(\zeta h_1 P)$ must hold. Thus,
$R_{\rm PS}(P,\infty)$ is upper bounded by maximizing $\frac{1}{N}\sum\limits_{n=1}^{N}\log_2\left(1+\frac{(1-\rho)h_{n}p_n}{\N}\right)$ subject to $\rho\geq \overline{E}/(\zeta h_1 P)$ and $\sum\limits_{n=1}^{N}p_n\leq P$.
Let $p_n'=(1-\rho)p_n,\forall n$, the above problem is then equivalent to maximizing $\frac{1}{N}\sum\limits_{n=1}^{N}\log_2\left(1+\frac{h_{n}p_n'}{\N}\right)$
subject to $\rho\geq \overline{E}/(\zeta h_1 P)$ and
$\sum\limits_{n=1}^{N}p_n'\leq (1-\rho)P$. Since $\rho\geq \overline{E}/(\zeta h_1 P)$, it follows that $(1-\rho)P\leq P-\overline{E}/(\zeta h_1)$. Note that according to Proposition \ref{proposition-1}, $R_{\rm TS}(P,\infty)$ is obtained (with $\alpha_1\rightarrow1$) by maximizing $\frac{1}{N}\sum\limits_{n=1}^{N}\log_2\left(1+\frac{h_{n}p_{1,n}}{\N}\right)$ subject to $\sum\limits_{n=1}^{N}p_{1,n}\leq P-\overline{E}/(\zeta h_1)$.
Therefore, we have $R_{\rm PS}(P,\infty)\leq R_{\rm TS}(P,\infty)$.
\end{proof}

For the other extreme case when $\Peak=P/N$, we have the following proposition.
\begin{proposition}\label{proposition-3}
In the case of a single-user OFDM-based SWIPT system with $\Peak=P/N$, the maximum rate achieved by the TS scheme is no larger than that achieved by the PS scheme, i.e., $R_{\rm TS}(P,P/N)\leq R_{\rm PS}(P,P/N)$.
\end{proposition}
\begin{proof}
With $\Peak=P/N$, the total power constraint is redundant for both TS and PS. Thus, the optimal power allocation for TS is given by $p_{1,n}^\ast=p_{2,n}^\ast=\Peak,\forall n$.
It follows that $\alpha_2\geq \frac{\overline{E}}{\zeta \Peak \sum\limits_{n=1}^N h_n}$. Then we have the optimal $\alpha_1^\ast=1-\frac{\overline{E}}{\zeta \Peak \sum\limits_{n=1}^N h_n}$.
$R_{\rm TS}(P,P/N)$ is thus given by $\frac{\alpha_1^\ast}{N}\sum\limits_{n=1}^N\log_2\left(1+\frac{h_n \Peak}{\N}\right)$.
On the other hand, the optimal power allocation for PS is given by $p_{n}^\ast=\Peak,\forall n$. It follows that $\rho^\ast=\frac{\overline{E}}{\zeta \Peak \sum\limits_{n=1}^N h_n}=1-\alpha_1^\ast$.
$R_{\rm PS}(P,P/N)$ is thus given by $\frac{1}{N}\sum\limits_{n=1}^N\log_2\left(1+\frac{\alpha_1^\ast h_n \Peak}{\N}\right)$. Due to the concavity of the logarithm function,
we have $R_{\rm TS}(P,P/N)\leq R_{\rm PS}(P,P/N)$, which completes the proof.
\end{proof}

In fact, from the proof of Proposition \ref{proposition-3}, we have $R_{\rm TS}(P,P/N)\leq R_{\rm PS}(P,P/N)$ provided that
equal power allocation (not necessarily equals to $\Peak$) over all SCs are employed for both TS and PS schemes.
Note that for a single-user OFDM-based SWIPT system with $P/N< \Peak<\infty$, the performance comparison between the TS scheme and the PS scheme remains unknown analytically.
From Proposition \ref{proposition-2} and Proposition \ref{proposition-3},
neither the TS scheme nor the PS scheme is always better. It suggests that for a single-user OFDM-based SWIPT system with sufficiently small peak power, the PS scheme may be better; with sufficiently large peak power, the TS scheme may be better.

For the special case that $N=1$, i.e., a single-carrier point-to-point SWIPT system, the following proposition shows that: for $\Peak\rightarrow\infty$, the TS and PS schemes achieve the same rate; for a finite peak power $P/N\leq\Peak<\infty$, the TS scheme performs no better than the PS scheme.

\begin{proposition}\label{proposition-4}
In the case of a single-carrier point-to-point SWIPT system with $N=1$, we have $R_{\rm TS}(P,\Peak)\leq R_{\rm PS}(P,\Peak)$, with equality if $\Peak\rightarrow\infty$.
\end{proposition}
\begin{proof}
Since $N=1$, we remove the SC index $n$ in the subscripts of $h_n$, $p_{1,n}$, $p_{2,n}$ and $p_n$.
For the PS scheme, to satisfy the EH constraint, we have $\rho\geq\overline{E}/(\zeta hP)$; thus, with $\rho=\overline{E}/(\zeta hP)$, the maximum rate by the PS scheme is given by $R_{\rm PS}(P,\Peak)=\log_2\left(1+\frac{hP-\overline{E}/\zeta}{\N}\right)$.
For the TS scheme, we have $\alpha_2p_2\geq \overline{E}/(\zeta h)$ to satisfy the EH constraint. It follows that $\alpha_1p_1\leq P-\overline{E}/(\zeta h)$. Therefore, $R_{\rm TS}(P,\Peak)\leq\alpha_1\log_2\left(1+\frac{hP-\overline{E}/\zeta}{\alpha_1 \N}\right)\leq R_{\rm PS}(P,\Peak)$, and the
equality holds if $\alpha_1\rightarrow1$. By Proposition \ref{proposition-1}, $R_{\rm TS}(P,\infty)$ is achieved by
$\alpha_1\rightarrow1$; thus, the above equality holds if $\Peak\rightarrow\infty$, which completes the proof.
\end{proof}

Figs. \ref{fig:Rate-K1} and \ref{fig:Rate-K1N1} show the achievable rates by
different schemes versus different minimum required harvested energy $\overline{E}$.
For Fig. \ref{fig:Rate-K1}, the total bandwidth is assumed
to be $10$MHz, which is equally divided as $N=64$ subcarriers.
The six-tap exponentially distributed power profile is used to generate the frequency-selective fading channel.
For Fig. \ref{fig:Rate-K1N1} with $N=1$, i.e., a single-carrier point-to-point SWPT system,
the bandwidth is assumed to be $1$MHz.
For both figures, the transmit power is assumed to be $1$watt(W) or 30dBm.
The distance from the transmitter to the receiver is 1 meter(m), which results in $-30$dB path-loss for all the channels at a carrier frequency $900\rm{MHz}$ with path-loss exponent equal to 3.
For the energy receivers, it is assumed that $\zeta=0.2$.
For the information receivers, the noise spectral density is assumed to be $-112$dBm/Hz.
The MCS gap is assumed to be $\Gamma=9$dB.

\begin{figure}
\begin{center}
\scalebox{0.6}{\includegraphics{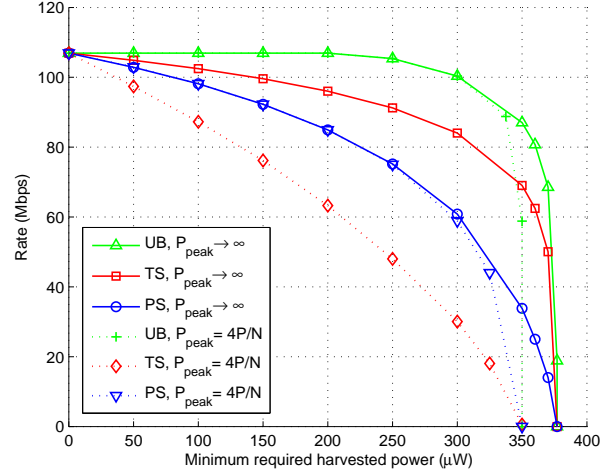}}
\end{center}
\caption{Achievable rate versus minimum required harvested energy in a single-user OFDM-based SWIPT system, where $N=64$.}
\label{fig:Rate-K1}
\end{figure}

\begin{figure}
\begin{center}
\scalebox{0.6}{\includegraphics{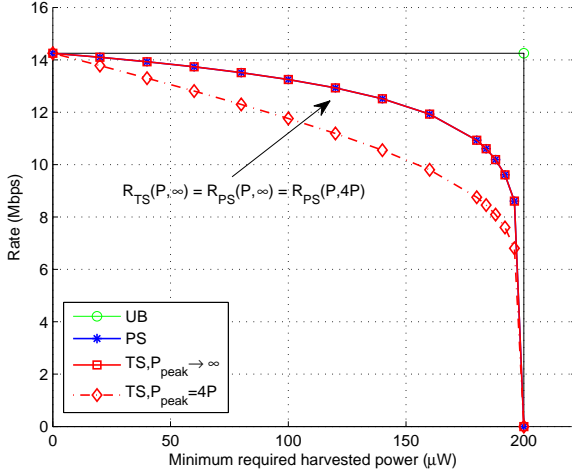}}
\end{center}
\caption{Achievable rate versus minimum required harvested energy in a single-carrier point-to-point SWIPT system, where $N=1$.}
\label{fig:Rate-K1N1}
\end{figure}

In both Fig. \ref{fig:Rate-K1} and Fig. \ref{fig:Rate-K1N1},
it is observed that for both TS and PS schemes, the achievable rate decreases as the minimum required
harvested energy $\overline{E}$ increases, since the available energy for information decoding decreases as $\overline{E}$ increases.
In Fig. \ref{fig:Rate-K1} with $N=64$, it is observed that there is a significant gap between the achievable rate by TS with $\Peak=4P/N$ and that by TS with $\Peak\rightarrow\infty$; moreover, the gap increases as $\overline{E}$ increases. This is because that with
$\Peak\rightarrow\infty$, all transmission time can be utilized for information decoding by letting $\alpha_1\rightarrow1$ (c.f. Proposition \ref{proposition-1}); whereas for a finite $\Peak=4P/N$, a nonzero
transmission time needs to be scheduled for energy harvesting.
For the PS scheme, this performance gap due to finite peak power constraint is only observed when $\overline{E}$ is sufficiently large.
Comparing the TS and PS schemes in Fig. \ref{fig:Rate-K1}, it is observed that TS outperforms PS when $\Peak\rightarrow\infty$; however, for sufficiently small $\Peak$, e.g., $\Peak=4P/N$, PS outperforms TS.
In Fig. \ref{fig:Rate-K1N1} with $N=1$, it is observed that when $\Peak\rightarrow\infty$, the achievable rate by the TS scheme is the same as that by the PS scheme; when $\Peak=4P$, the achievable rate by the TS scheme is no larger than that by the PS scheme, which is in accordance with Proposition \ref{proposition-4}.

\section{Resource Allocation in a Multiuser System}\label{sec:solution}
In this section, we consider the general case of an OFDM-based SWIPT system with multiple users.
We derive the optimal transmission strategies for the two schemes proposed in Section \ref{sec:system model}, and compare their performances.

\subsection{Time Switching}\label{sec:solution-TS}
We first reformulate Problem (P-TS) by introducing a set of new non-negative variables: $q_{k,n}=\alpha_k p_{k,n},k=1,\ldots,K+1,n=1,\ldots,N$. Moreover, we define $\alpha_k\log_2 \left(1+\frac{h_{k,n}q_{k,n}}{\N \alpha_k}\right)=0$ at $\alpha_k=0$ to keep continuity at $\alpha_k=0$.
(P-TS) is thus equivalent to the following problem:\footnote{Similar to the single-user system (c.f. Remark \ref{remark:Peak-infty}), for the case of $\Peak\rightarrow\infty$, we allow $\alpha_k\rightarrow0$ while $q_{k,n}>0$ by letting $p_{k,n}\rightarrow\infty$.}

\begin{small}
\vspace{-0.1in}
\begin{align}
~\mathop{\mathtt{max.}}_{\{q_{k,n}\},\{\alpha_k\}} & ~~ \frac{1}{N}\sum\limits_{k=1}^{K}\sum\limits_{n=1}^{N} w_k\alpha_k \log_2 \left(1+\frac{h_{k,n}q_{k,n}}{\N \alpha_k}\right) \nonumber \\
\mathtt{s.t.} & ~~ \zeta\sum\limits_{i\neq k}^{K+1}\sum\limits_{n=1}^{N} h_{k,n} q_{i,n} \geq \Emin,
    ~ ~ k=1,\ldots,K, \nonumber\\
              & ~~ \sum\limits_{k=1}^{K+1}\sum\limits_{n=1}^{N}q_{k,n}\leq P, \nonumber\\
              & ~~ 0\leq q_{k,n}\leq \alpha_k \Peak,k=1,\ldots,K+1,\forall n, \nonumber \\
              & ~~ \sum\limits_{k=1}^{K+1} \alpha_k \leq 1, ~ 0\leq\alpha_k\leq 1,k=1,\ldots,K+1.  \label{prob:TS-1}
\end{align}
\end{small}

After finding the optimal $\{q_{k,n}^\ast\}$ and $\{\alpha_k^\ast\}$ for Problem (\ref{prob:TS-1}), the optimal power allocation $\{p_{k,n}^\ast\}$ for Problem (P-TS) is given by $p_{k,n}^\ast=q_{k,n}^\ast/\alpha_k^\ast,k=1,\ldots,K+1,n=1,\ldots,N$ provided that $\alpha_k^\ast>0$.
From the constraint $0\leq q_{k,n}\leq\alpha_k \Peak,k=1,\ldots,K+1,n=1,\ldots,N$, we have
$q_{k,n}=0$ if $\alpha_k=0$ and $\Peak<\infty$. Thus, if $\alpha_k^\ast=0,k=1,\ldots,K+1$ and $\Peak<\infty$, the allocated power will be $p_{k,n}^\ast=0,n=1\ldots,N$, since no information/power transmission is scheduled at slot $k$.
For the extreme case of $\Peak\rightarrow\infty$, if $q_{k,n}^\ast=0,\alpha_k^\ast=0,k=1,\ldots,K+1,n=1,\ldots,N$, then the allocated power will be $p_{k,n}^\ast=0$; if $q_{k,n}^\ast>0$ and $\alpha_k^\ast=0$, then we have $p_{k,n}^\ast\rightarrow\infty$.

\begin{lemma}\label{lemma:1}
Function $f(q_{k,n},\alpha_k)$ is jointly concave in $\alpha_k\geq0$ and $q_{k,n}\geq0$, where
\begin{equation}\label{}
f(q_{k,n},\alpha_k)=\begin{cases}
	\alpha_k \log_2 \left(1+\frac{h_{k,n}q_{k,n}}{\N \alpha_k}\right), & \alpha_k>0, \\
	0, & \alpha_k=0.
\end{cases}
\end{equation}
\end{lemma}
\begin{proof}
Please refer to Appendix \ref{appendix:proof 1}.
\end{proof}

From Lemma \ref{lemma:1}, as a non-negative weighted sum of $f(q_{k,n},\alpha_k)$, the new objective function of Problem (\ref{prob:TS-1}) is
jointly concave in $\{\alpha_k\}$ and $\{q_{k,n}\}$. Since the constraints are now all affine, Problem (\ref{prob:TS-1}) is convex, and thus can be optimally solved by applying the Lagrange duality method, as will be shown next.

The Lagrangian of Problem (\ref{prob:TS-1}) is given by
\begin{small}
\begin{align}\label{eq:Lagrangian-TS}
& \mathcal{L}\left(\{q_{k,n}\},\{\alpha_k\},\{\lambda_i\},\mu,\nu\right) \nonumber\\
   & ~~ =\frac{1}{N}\sum\limits_{k=1}^{K}\sum\limits_{n=1}^{N}w_k\alpha_k\log_2\left(1+\frac{h_{k,n}q_{k,n}}{\N \alpha_k}\right) \nonumber\\
   & ~~~~   + \sum\limits_{i=1}^{K}\lambda_i\left(\zeta\sum\limits_{k\neq i}^{K+1}\sum\limits_{n=1}^{N}h_{i,n}q_{k,n}-\overline{E}_i\right)   \nonumber\\
    & ~~~~   + \mu\left(P-\sum\limits_{k=1}^{K+1}\sum\limits_{n=1}^{N}q_{k,n}\right)
        + \nu\left(1-\sum\limits_{k=1}^{K+1}\alpha_k\right)
\end{align}
\end{small}%
where $\lambda_i,i=1,\ldots,K$, $\mu$, and $\nu$ are the non-negative dual variables associated with the corresponding constraints in (\ref{prob:TS-1}). The dual function
$g\left(\{\lambda_i\},\mu,\nu\right)$ is then defined as the optimal value of the following problem.
\begin{align}\label{eq:dual func-TS}
~\mathop{\mathtt{max.}}_{\{q_{k,n}\},\{\alpha_k\}} & ~~ \mathcal{L}\left(\{q_{k,n}\},\{\alpha_k\},\{\lambda_i\},\mu,\nu\right) \nonumber \\
\mathtt{s.t.} & ~~ 0\leq q_{k,n}\leq \alpha_k \Peak, ~ k=1,\ldots,K+1,\forall n,\nonumber \\
              & ~~ 0\leq\alpha_k\leq 1,~~ k=1,\ldots,K+1.
\end{align}
The dual problem is thus defined as $\min_{\{\lambda_i\},\mu,\nu} g\left(\{\lambda_i\},\mu,\nu\right)$.

First, we consider the maximization problem in (\ref{eq:dual func-TS}) for obtaining $g\left(\{\lambda_i\},\mu,\nu\right)$ with a given set of $\{\lambda_i\}$, $\mu$, and $\nu$.
We define $\mathcal{L}_k,k=1,\ldots,K+1$, as shown in (\ref{eq:Lk}) at the top of this page.
\begin{figure*}[ht]
\begin{small}
\begin{align}\label{eq:Lk}
  \mathcal{L}_{k}:=\begin{cases}
  \frac{w_k\alpha_k}{N}\sum\limits_{n=1}^{N}\log_2\left(1+\frac{h_{k,n}q_{k,n}}{\N \alpha_k}\right)
    +\zeta\sum\limits_{i\neq k}^{K}\lambda_i \sum\limits_{n=1}^{N}h_{i,n}q_{k,n}
    -\mu\sum\limits_{n=1}^{N}q_{k,n}-\nu \alpha_k, & k=1,\ldots,K  \\
  \zeta\sum\limits_{i=1}^{K}\lambda_i \sum\limits_{n=1}^{N}h_{i,n}q_{k,n}
    -\mu\sum\limits_{n=1}^{N}q_{k,n}-\nu \alpha_k, & k=K+1.
           \end{cases}
\end{align}
\end{small}\hrulefill
\end{figure*}
Then for the Lagrangian in (\ref{eq:Lagrangian-TS}), we have
\begin{equation}\label{}
  \mathcal{L}=\sum\limits_{k=1}^{K+1}\mathcal{L}_k-\sum\limits_{i=1}^{K}\lambda_i \overline{E}_i +\mu P+\nu.
\end{equation}
Thus, for each given $k$, the maximization problem in (\ref{eq:dual func-TS}) can be decomposed as
\begin{align}\label{prob:TS-Lk}
~\mathop{\mathtt{max.}}_{\{q_{k,n}\},\alpha_k} & ~~ \mathcal{L}_{k} \nonumber\\
\mathtt{s.t.} & ~~ 0\leq q_{k,n}\leq \alpha_k \Peak, ~ n=1,\ldots,N \nonumber\\
              & ~~ 0\leq\alpha_k\leq 1.
\end{align}

We first study the solution for Problem  (\ref{prob:TS-Lk}) with given $k=1,\ldots,K$.
From (\ref{eq:Lk}), we have
\begin{equation}\label{eq:diff Lk}
  \frac{\partial \mathcal{L}_{k}}{\partial q_{k,n}} =
    \frac{w_k\alpha_k h_{k,n}}{N(\N \alpha_k+h_{k,n}q_{k,n})\ln 2}
        +\zeta\sum\limits_{i\neq k}^{K}\lambda_i h_{i,n}-\mu,  \forall n.
\end{equation}
Given $\alpha_k,k=1,\ldots,K$, the $q_{k,n},n=1,\ldots,N$ that maximizes $\mathcal{L}_{k}$ can be obtained by
setting $\frac{\partial \mathcal{L}_{k}}{\partial q_{k,n}}=0$ to give

\begin{small}
\begin{equation}\label{eq:PA-TS}
      q_{k,n} = \alpha_k \min\left(\left(\frac{w_k}{N\left(\mu-\zeta\sum\limits_{i\neq k}^{K}\lambda_i h_{i,n}\right)\ln 2} -\frac{\N }{h_{k,n}}\right)^+,\Peak \right)
\end{equation}
\end{small}%
where $(x)^+\triangleq\max(0,x)$.
For given $\{q_{k,n}\}$, it appears that there is no closed-form expression for the optimal $\alpha_k$ that maximizes $\mathcal{L}_k$. However, since $\mathcal{L}_k$ is a concave function of $\alpha_k$ with given $\{q_{k,n}\}$, $\alpha_k$ can be found numerically by a simple bisection search over $0\leq\alpha_k\leq1$.
To summarize, for given $k=1,\ldots,K$, Problem (\ref{prob:TS-Lk}) can be solved by iteratively optimizing between $\{q_{k,n}\}$ and $\alpha_k$ with one of them fixed at one time, which is known as block-coordinate descent method \cite{Rich}.

Next, we study the solution for Problem  (\ref{prob:TS-Lk}) for $k=K+1$, i.e., for the power slot, which is a LP. Define the set $\mathcal{N}_1=\left\{n\in\mathcal{N}: \zeta\sum\limits_{i=1}^K \lambda_i h_{i,n}-\mu>0\right\}$.
From (\ref{eq:Lk}), to maximize $\mathcal{L}_{K+1}$ we have
\begin{align}\label{eq:q_K+1,n}
q_{K+1,n}=\begin{cases}
	\alpha_{K+1} \Peak, & n\in\mathcal{N}_1, \\
	0, & n\in\mathcal{N}\backslash\mathcal{N}_1
\end{cases}
\end{align}
and
\begin{align}\label{eq:alpha-K+1}
\alpha_{K+1}=\begin{cases}
	1, & {\rm if} \sum\limits_{n\in\mathcal{N}_1}\left(\zeta\sum\limits_{i=1}^K\lambda_i h_{i,n}-\mu\right)\Peak -\nu >0, \\
	0, & {\rm otherwise}.
\end{cases}
\end{align}

After obtaining $g\left(\{\lambda_i\},\mu,\nu\right)$ with given $\{\lambda_i\}$, $\mu$, and $\nu$, the minimization of $g\left(\{\lambda_i\},\mu,\nu\right)$ over $\{\lambda_i\}$, $\mu$, and $\nu$ can be efficiently solved by the ellipsoid method \cite{Boydnote}. A subgradient of this problem required for the ellipsoid method is provided by the following proposition.

\begin{proposition}\label{proposition-subgradient}
For Problem (\ref{prob:TS-1}) with a dual function $g\left(\{\lambda_i\},\mu,\nu\right)$, the following choice of $\mv{d}$
is a subgradient for $g\left(\{\lambda_i\},\mu,\nu\right)$:
\begin{align}\label{eq:subgradient-TS}
d_i=\begin{cases}
	\zeta\sum\limits_{k\neq i}^{K+1}\sum\limits_{n=1}^{N}h_{i,n}\dot{q}_{k,n}-\overline{E}_i, &  ~~ i=1,\ldots,K, \\
	P-\sum\limits_{k=1}^{K+1}\sum\limits_{n=1}^{N}\dot{q}_{k,n}, & ~~ i=K+1,\\
    1-\sum\limits_{k=1}^{K+1}\dot{\alpha}_k, & ~~ i=K+2.
\end{cases}
\end{align}
where $\{\dot{q}_{k,n}\}$ and $\{\dot{\alpha}_k\}$ is the solution of the maximization problem (\ref{eq:dual func-TS})
with given $\{\lambda_i\},\mu$ and $\nu$.
\end{proposition}
\begin{proof}
Please refer to Appendix \ref{appendix:proof 2}.
\end{proof}

Note that the optimal $q_{k,n}^\ast,k=1\ldots,K,n=1,\ldots,N$ and $\alpha_k^\ast,k=1,\ldots,K$ are obtained at optimal
$\{\lambda_i^\ast\}$, $\mu^\ast$, and $\nu^\ast$.
Given $\{q_{k,n}\}$, the objective function in Problem (\ref{prob:TS-1}) is an increasing function of $\alpha_k,k=1,\ldots,K$. Thus, the optimal $\alpha_k^\ast$'s, $k=1,\ldots,K+1$ satisfy
$\sum\limits_{k=1}^{K+1}\alpha_k^\ast=1$; otherwise, the objective can be improved by increasing some of the $\alpha_k$'s, $k=1,\ldots,K$. Then, the optimal $\alpha_{K+1}$ is given by $\alpha_{K+1}^\ast = 1-\sum\limits_{k=1}^K \alpha_k^\ast$.
With $\alpha_k=\alpha_k^\ast,k=1,\ldots,K+1$, $q_{k,n}=q_{k,n}^\ast,k=1,\ldots,K,n=1,\ldots,N$, Problem (\ref{prob:TS-1}) becomes a LP with variables $\{q_{K+1,n}\}$. The optimal values of $\{q_{K+1,n}\}$ are obtained by solving this LP.

To summarize, one algorithm to solve (P-TS) is given in Table \ref{table1}.
For the algorithm given in Table \ref{table1}, the computation time is dominated by the ellipsoid method in steps I)-III) and the LP in step V). In particular, the time complexity of steps 1)-3) is of order $K^2N$, step 4) is of order $N$, step 5) is of order $K^2N$, and step 6) is of order $K^2$. Thus, the time complexity of steps 1)-6) is of order $K^2N$, i.e., $\mathcal{O}(K^2N)$. Note that step II) iterates $\mathcal{O}(K^2)$ to converge \cite{Boydnote}, thus the time complexity of steps I)-III) is $\mathcal{O}(K^4N)$. The time complexity of the LP in step V) is $\mathcal{O}(KN^2+N^3)$ \cite{Boydbook}. Therefore, the time complexity of the algorithm in Table \ref{table1} is $\mathcal{O}(K^4N+KN^2+N^3)$.
\begin{table}[htp]
\begin{center}
\caption{Algorithm for Solving Problem (P-TS)}
\vspace{0.1cm}
\hrule
\vspace{0.1cm}
\begin{itemize}
\item[I)] Initialize $\{\lambda_i>0\}$, $\mu>0$ and $\nu>0$.
\item[II)] {\bf repeat}
    \begin{itemize}
    \item[1)] Initialize $\alpha_k=1/K,k=1,\ldots,K$.
    \item[2)] {\bf repeat}
        \begin{itemize}
            \item [a)] For $k=1,\ldots,K$, compute $\{q_{k,n}\}$ by (\ref{eq:PA-TS}).
            \item [b)] For $k=1,\ldots,K$, obtain $\alpha_k$ that maximizes $\mathcal{L}_k$ with given $\{q_{k,n}\}$ by bisection search.
        \end{itemize}
    \item[3)] {\bf until} improvement of $\mathcal{L}_k,k=1,\ldots,K$ converges to a prescribed accuracy.
    \item[4)] Compute $\{q_{K+1,n}\}$ and $\alpha_{K+1}$ by (\ref{eq:q_K+1,n}) and (\ref{eq:alpha-K+1}).
    \item[5)] Compute the subgradient of $g(\{\lambda_i\},\mu,\nu)$ by (\ref{eq:subgradient-TS}).
    \item[6)] Update $\{\lambda_i\}$, $\mu$ and $\nu$ according to the ellipsoid method.
    \end{itemize}
\item[III)] {\bf until} $\{\lambda_i\}$, $\mu$ and $\nu$ converge to a prescribed accuracy.
\item[IV)] Set $q_{k,n}^\ast=q_{k,n},k=1,\ldots,K,n=1,\ldots,N$, $\alpha_k^\ast=\alpha_k,k=1,\ldots,K$ and $\alpha_{K+1}^\ast=1-\sum\limits_{k=1}^{K}\alpha_k^\ast$.
\item[V)] Obtain $q_{K+1,n}^\ast,n=1,\ldots,N$ by solving Problem (\ref{prob:TS-1}) with $\alpha_k=\alpha_k^\ast,k=1,\ldots,K+1$, $q_{k,n}=q_{k,n}^\ast,k=1,\ldots,K,n=1,\ldots,N$.
\item[VI)] For $k=1,\ldots,K+1$ and $n=1,\ldots,N$, if $\alpha_k^\ast>0$, set $p_{k,n}^\ast=q_{k,n}^\ast/\alpha_k^\ast$; if $\alpha_k^\ast=0$ and $q_{k,n}^\ast=0$, set $p_{k,n}^\ast=0$; if $\alpha_k^\ast=0$ and $q_{k,n}^\ast>0$, set $p_{k,n}^\ast\rightarrow\infty$.
\end{itemize}
\vspace{0.1cm} \hrule\label{table1}
\end{center}
\end{table}

Similar with the single-user case, we have the following proposition.
\begin{proposition}
In the case of a multiuser OFDM-based SWIPT system with $K\geq2$ and $\Peak\rightarrow\infty$, the maximum rate by the TS scheme, i.e., $R_{\rm TS}(P,\infty)$, is achieved by $\alpha_{K+1}=0$ or $\alpha_{K+1}\rightarrow 0$.
\end{proposition}
\begin{proof}
In the equivalent Problem (\ref{prob:TS-1}) with $\Peak\rightarrow\infty$, the EH constraints and the total power constraint are independent of $\alpha_k,k=1,\ldots,K+1$. The objective in Problem (\ref{prob:TS-1}) is an increasing function of $\alpha_k,k=1,\ldots,K$ for given $\{q_{k,n}\}$. Thus, the maximum achievable rate is obtained by minimizing the time allocated to the power slot, i.e., $\alpha_{K+1}=0$ (when $q_{K+1,n}^\ast=0,\forall n$) or $\alpha_{K+1}\rightarrow 0$ (when $q_{K+1,n}^\ast>0$ for some $n$).
\end{proof}

It is worth noting that for the multiuser system with $K\geq2$ and $\Peak\rightarrow\infty$, it is possible that the maximum rate by the TS scheme is
achieved by $\alpha_{K+1}=0$, in which case no additional power slot is scheduled and all users simply harvest energy
at the slots scheduled for other users to transmit information. In contrast,
for the single-user $K=1$ case, the power slot is always needed if $\overline{E}>0$.

\begin{remark}\label{remark:conventional TDMA}
In Problem (\ref{prob:TS-1}), when $K\geq2$ and $\Emin=0,k=1,\ldots,K$, the system becomes a conventional TDMA system without EH constraints.
Assume that the harvesting energy at each user by the optimal transmission strategy for this system is given by $E_k^{\rm th},k=1,\ldots,K$.
Then for a system with $0\leq \Emin \leq E_k^{\rm th},k=1,\ldots,K$, the same rate as that by the conventional TDMA system can be achieved.
\end{remark}

\subsection{Power Splitting}\label{sec:solution-PS}
Since Problem (P-PS) is non-convex, the optimal solution may be computationally difficult to obtain.
Instead, we propose a suboptimal algorithm to this problem by
iteratively optimizing $\{p_n\}$ and $\{\Pi(n)\}$ with fixed $\{\rho_k\}$, and optimizing $\{\rho_k\}$
with fixed $\{p_n\}$ and $\{\Pi(n)\}$.

Note that (P-PS) with given $\{p_n\}$ and $\{\Pi(n)\}$ is a convex problem, of which the objective function is a nonincreasing function of $\rho_k,\forall k$. Thus, the optimal power splitting ratio solution for (P-PS) with a given set of feasible $\{p_n\}$ and $\{\Pi(n)\}$ is obtained as
\begin{equation}\label{eq:rho}
  \rho_k=\frac{\Emin}{\zeta\sum\limits_{n=1}^{N} h_{k,n} p_n}, ~~ k=1,\ldots,K.
\end{equation}

Next, consider (P-PS) with a given set of feasible $\rho_k$'s, i.e.,
\begin{align}
~\mathop{\mathtt{max.}}_{\{p_n\},\{\Pi(n)\}} & ~~ \frac{1}{N}\sum\limits_{n=1}^{N}w_{\Pi(n)}\log_2 \left(1+\frac{h_{\Pi(n),n}^{\rm ID}p_n}{\N}\right) \nonumber \\
\mathtt{s.t.} & ~~ \zeta\sum\limits_{n=1}^{N} \Heh p_n \geq \Emin, ~~ k=1,\ldots,K,   \nonumber\\
              & ~~ \sum\limits_{n=1}^{N} p_n \leq P, ~~ 0\leq p_n\leq \Peak,n=1,\ldots,N  \label{prob:PS-rho}
\end{align}
where $\Hid \triangleq (1-\rho_k)h_{k,n},\forall k,n$ and $\Heh \triangleq\rho_k h_{k,n},\forall k,n$ can be viewed as the equivalent channel power gains for the information and energy receivers, respectively.
The problem in (\ref{prob:PS-rho}) is non-convex, due to the integer SC allocation $\Pi(n)$.
However, it has been shown that the duality gap of a similar problem to (\ref{prob:PS-rho}) without the harvested energy constrains converges to zero as the number of SCs, $N$, increases to infinity \cite{Cioffi,Yu}.
Thus, we solve Problem (\ref{prob:PS-rho}) by applying the Lagrange duality method assuming that it has a zero duality gap.\footnote{In our simulation setup considered in Section \ref{sec:numerical} with $N=64$, the duality gap of Problem (\ref{prob:PS-rho}) is observed to be negligibly small and thus can be ignored.}

The Lagrangian of Problem (\ref{prob:PS-rho}) is given by

\begin{small}\vspace{-0.1in}
\begin{align}\label{eq:Lagrangian}
&\mathcal{L}\left(\{p_n\},\{\Pi(n)\},\{\lambda_k\},\mu\right)= \nonumber\\
 &   \frac{1}{N}\sum\limits_{n=1}^{N}w_{\Pi(n)}\log_2 \left(1+\frac{h_{\Pi(n),n}^{\rm ID}p_n}{\N}\right) \nonumber\\
   &     + \sum\limits_{k=1}^{K}\lambda_k\left( \zeta\sum\limits_{n=1}^{N} \Heh p_n-\Emin\right)
        + \mu\left( P-\sum\limits_{n=1}^{N} p_n \right)
\end{align}
\end{small}%
where $\lambda_k$'s and $\mu$ are the non-negative dual variables associated with the corresponding constraints in (\ref{prob:PS-rho}).
The dual function is then defined as
\begin{equation}\label{eq:dual func}
   g\left(\{\lambda_k\},\mu\right)=\max\limits_{\{p_n\},\{\Pi(n)\}} \mathcal{L}\left(\{p_n\},\{\Pi(n)\},\{\lambda_k\},\mu\right).
\end{equation}
The dual problem is thus given by $\min_{\{\lambda_k\},\mu} g\left(\{\lambda_k\},\mu\right)$.

Consider the maximization problem in (\ref{eq:dual func}) for obtaining $g\left(\{\lambda_k\},\mu\right)$ with
a given set of $\{\lambda_k\}$ and $\mu$. For each given SC $n$, the maximization problem in (\ref{eq:dual func}) can be decomposed as
\begin{align}
~\mathop{\mathtt{max.}}_{p_n,\Pi(n)} & ~~ \mathcal{L}_n:=\frac{w_{\Pi(n)}}{N}\log_2 \left(1+\frac{h_{\Pi(n),n}^{\rm ID}p_n}{\N}\right) \nonumber\\
& ~~ ~~ ~~ ~~ +\zeta\sum\limits_{k=1}^{K}\lambda_k \Heh p_n -\mu p_n   \nonumber\\
\mathtt{s.t.} & ~~ 0\leq p_n\leq \Peak.  \label{eq:Ln}
\end{align}
Note that for the Lagrangian in (\ref{eq:Lagrangian}), we have
\begin{equation}\label{}
  \mathcal{L}=\sum\limits_{n=1}^{N}\mathcal{L}_n-\sum\limits_{k=1}^{K}\lambda_k \Emin +\mu P.
\end{equation}
From (\ref{eq:Ln}), we have
\begin{equation}\label{eq:differential L}
  \frac{\partial \mathcal{L}_n}{\partial p_n} = \frac{w_{\Pi(n)}h_{\Pi(n),n}^{\rm ID}}{N(\N+h_{\Pi(n),n}^{\rm ID}p_n)\ln 2}
        + \zeta\sum\limits_{k=1}^{K}\lambda_k \Heh -\mu.
\end{equation}
Thus, for any given SC allocation $\Pi(n)$, the optimal power allocation for Problem (\ref{eq:Ln}) is obtained as

\begin{footnotesize}
\begin{equation}\label{eq:PA}
      p_n^\ast(\Pi)=\min\left(\left(\frac{w_{\Pi(n)}}{N\left(\mu-\zeta\sum\limits_{k=1}^{K}\lambda_k \Heh \right)\ln 2} -\frac{\N}{h_{\Pi(n),n}^{\rm ID}}\right)^+,\Peak\right).
\end{equation}
\end{footnotesize}%
Thus, for each given $n$, the optimal SC allocation $\Pi^\ast(n)$ to maximize $\mathcal{L}_n$ can be obtained, which is shown in (\ref{eq:subcarrier}) at the top of next page.
\begin{figure*}[ht]
\begin{small}
\begin{align}\label{eq:subcarrier}
\Pi^\ast(n)=\arg\max\limits_{\Pi(n)}  \left( \frac{w_{\Pi(n)}}{N} \left( \log_2\left(\frac{w_{\Pi(n)}h_{\Pi(n),n}^{\rm ID}}{N\N \left(\mu-\zeta\sum\limits_{k=1}^{K}\lambda_k \Heh \right)\ln 2}\right)\right)^+
-\left(\frac{w_{\Pi(n)}}{N\ln 2}-\frac{\N \left(\mu-\zeta\sum\limits_{k=1}^{K}\lambda_k \Heh \right)}{h_{\Pi(n),n}^{\rm ID}}\right)^+ \right).
\end{align}
\end{small}\hrulefill
\end{figure*}
Note that (\ref{eq:subcarrier}) can be solved by exhaustive search over the user set $\{1,\ldots,K\}$.

After obtaining $g\left(\{\lambda_k\},\mu\right)$ with given $\{\lambda_k\}$ and $\mu$, the minimization of $g\left(\{\lambda_k\},\mu\right)$ over $\{\lambda_k\}$ and $\mu$ can be efficiently solved by the ellipsoid
method \cite{Boydnote}. A subgradient of this problem required for the ellipsoid method is provided by the following proposition.

\begin{proposition}
For Problem (\ref{prob:PS-rho}) with a dual function $g\left(\{\lambda_k\},\mu\right)$, the following choice of $\mv{d}$
is a subgradient for $g\left(\{\lambda_k\},\mu\right)$:
\begin{align}\label{eq:subgradient-PS}
d_k=\begin{cases}
	\zeta\sum\limits_{n=1}^{N} \Heh \dot{p}_n-\Emin, &  ~~ k=1,\ldots,K, \\
	P-\sum\limits_{n=1}^{N} \dot{p}_n , & ~~ k=K+1.
\end{cases}
\end{align}
where $\{\dot{p}_n\}$ is the solution of the maximization problem (\ref{eq:dual func}) with given $\{\lambda_k\}$ and $\mu$.
\end{proposition}
\begin{proof}
The proof is similar as the proof of Proposition \ref{proposition-subgradient}, and thus is omitted.
\end{proof}

\begin{remark}\label{remark:P-UB}
The optimal solution for (P-UB) can be obtained by setting $\Heh=\Hid=h_{k,n},\forall k,n$ in Problem (\ref{prob:PS-rho}).
Hence, the above developed solution is also applicable for Problem (P-UB).
\end{remark}

For (P-PS) with given $\{\rho_k\}$, the optimal $\{p_n\}$ and $\{\Pi(n)\}$ are obtained by
(\ref{eq:PA}) and (\ref{eq:subcarrier}), respectively.
Define the corresponding optimal value of Problem (\ref{prob:PS-rho}) as $R(\mv{\rho})$, where $\mv{\rho}=[\rho_1,\ldots,\rho_K]^T$.
Then $R(\mv{\rho})$ can be increased by optimizing $\rho_k$'s by (\ref{eq:rho}).
The above procedure can be iterated until $R(\mv{\rho})$ cannot be further improved.
Note that Problem (\ref{prob:PS-rho}) is guaranteed to be feasible at each iteration, provided that the initial $\rho_k$'s are feasible, since at each iteration we simply decrease $\rho_k$'s to make all the harvested energy constraints tight.
Thus, with given initial $\{\rho_k\}$, the iterative algorithm is guaranteed to converge to a local optimum of (P-PS) when all the harvested energy constraints in (\ref{prob:PS-rho}) are tight.

Note that the above local optimal solution depends on the choice of initial $\{\rho_k\}$.
To obtain a robust performance, we randomly generate $M$ feasible $\{\rho_k\}$ as the initialization steps, where $M$ is a sufficiently large number.\footnote{In general, as the number of users increases, the number of initialization steps needs to be increased to guarantee the robustness and optimality of the algorithm. However, large number of initialization steps increases the computation complexity, which may not be suitable for real-time applications.}
For each initialization step, the iterative algorithm is applied to obtain a local optimal solution for (P-PS). The final solution is selected as the one that achieves the maximum weighted sum-rate from all the solutions.

To summarize, the above iterative algorithm to solve (P-PS) is given in Table \ref{table3}.
For the algorithm given in Table \ref{table3}, the computation time is dominated by the ellipsoid method in steps A)-C). In particular, in step B), the time complexity of step a) is of order $KN$, step b) is of order $KN$, and step c) is of order $K^2$. Thus, the time complexity of steps a)-c) is of order $K^2+KN$, i.e., $\mathcal{O}(K^2+KN)$. Note that step B) iterates $\mathcal{O}(K^2)$ to converge \cite{Boydnote}, thus the time complexity of the ellipsoid method is $\mathcal{O}(K^4+K^3N)$. Considering further the number of initialization steps $M$, the time complexity of the algorithm in Table \ref{table3} is $\mathcal{O}(K^4M+K^3NM)$.
\begin{table}[htp]
\begin{center}
\caption{Iterative algorithm for Solving Problem (P-PS)}
\vspace{0.1cm}
\hrule
\vspace{0.1cm}
\begin{itemize}
\item[I)] Randomly generate $M$ feasible $\{\rho_k\}$ as different initialization steps.
\item[II)] For each initialization step:
\begin{itemize}
\item[1)] Initialize $\{\rho_k\}$.
\item[2)] {\bf repeat}
    \begin{itemize}
        \item [A)] Compute $\{\Heh\}$ and $\{\Hid\}$. Initialize $\{\lambda_k>0\}$ and $\mu>0$.
        \item [B)] {\bf repeat}
            \begin{itemize}
                \item [a)] Compute $\{p_n\}$ and $\{\Pi(n)\}$ by (\ref{eq:PA}) and (\ref{eq:subcarrier}) with given $\{\lambda_k\}$ and $\mu$.
                \item [b)] Compute the subgradient of $g(\{\lambda_k\},\mu)$ by (\ref{eq:subgradient-PS}).
                \item [c)] Update $\{\lambda_k\}$ and $\mu$ according to the ellipsoid method.
            \end{itemize}
        \item [C)] {\bf until} $\{\lambda_k\}$ and $\mu$ converge to a prescribed accuracy.
        \item [D)] Update $\{\rho_k\}$ by (\ref{eq:rho}) with fixed $\{p_n\}$ and $\{\Pi(n)\}$.
    \end{itemize}
\item[3)] {\bf until} $\left|\zeta\sum\limits_{n=1}^{N}\Heh p_n-\Emin\right|<\delta,\forall k$, where $\delta>0$ controls the algorithm accuracy.
\end{itemize}
\item[III)] Select the one that achieves the maximum weighted sum-rate from the $M$ solutions.
\end{itemize}
\vspace{0.1cm} \hrule\label{table3}
\end{center}
\end{table}

\subsection{Performance Comparison}\label{sec:numerical}
We provide simulation results under a practical system setup.
For each user, we use the same parameters as the single-user case with $N=64$ in Section \ref{sec:solution-single}. The channels for different users are generated independently.
In addition, it is assumed that $w_k=1,\forall k$, i.e., sum-rate maximization is considered.
The minimum harvested energy is assumed to be the same for all users, i.e., $\Emin=\overline{E},\forall k$.
The number of initialization steps $M$ is set to be 100.

\begin{figure}
\begin{center}
\scalebox{0.6}{\includegraphics{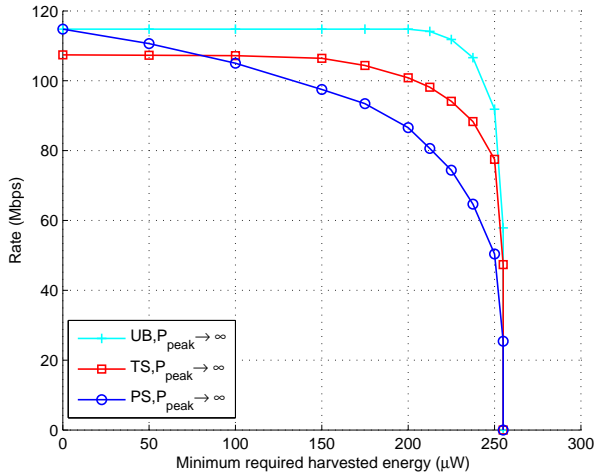}}
\end{center}
\caption{Achievable rates versus minimum required harvested energy in a multi-user OFDM-based SWIPT system, where $K=4$, $N=64$, and $\Peak\rightarrow\infty$.}
\label{fig:RE-K4-NoPeak}
\end{figure}

\begin{figure}
\begin{center}
\scalebox{0.6}{\includegraphics{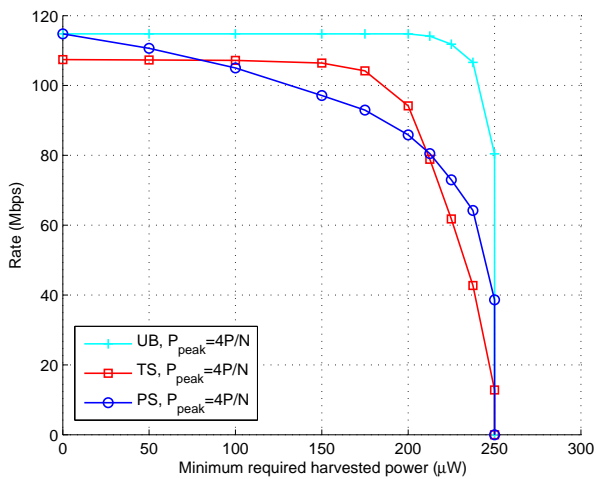}}
\end{center}
\caption{Achievable rates versus minimum required harvested energy in a multi-user OFDM-based SWIPT system, where $K=4$, $N=64$, and $\Peak=4P/N$.}
\label{fig:RE-K4-Peak}
\end{figure}

Figs. \ref{fig:RE-K4-NoPeak} and \ref{fig:RE-K4-Peak} show the achievable rates versus the minimum required harvested energy by different schemes with $K=4$. We assume $\Peak\rightarrow\infty$ in Fig. \ref{fig:RE-K4-NoPeak}, and
$\Peak=4P/N$ in Fig. \ref{fig:RE-K4-Peak}. Fig. \ref{fig:alpha} shows the time ratio of the EH slot versus minimum required harvested energy for the TS scheme in Fig. \ref{fig:RE-K4-Peak}.
In Fig. \ref{fig:RE-K4-NoPeak} with $\Peak \rightarrow\infty$, it is observed that when $\overline{E}>0$, the achievable rates by both TS and PS are less than the upper bound. For the TS scheme, the maximum rate is achieved when $\overline{E}$ is less than 150$\mu$W (c.f. Remark \ref{remark:conventional TDMA});
when $\overline{E}$ is larger than 150$\mu$W, the achievable rate decreases as $\overline{E}$ increases. For the PS scheme, the achievable rate decreases as $\overline{E}$ increases, since for larger $\overline{E}$ more power needs to be split for EH at each receiver.
Comparing the TS and PS schemes, it is observed that for sufficiently small $\overline{E}$ ($0\leq\overline{E}\leq80\mu$W), the achievable rate by PS is larger than that by TS.
This is because that when the required harvested energy is sufficiently small, only a small portion of power needs to be split for energy harvesting, and the PS scheme may take the advantage of the frequency diversity by subcarrier allocation.
For sufficiently large $\overline{E}$ ($80<\overline{E}\leq255\mu$W), it is observed that the achievable rate by TS is larger than that by PS.
In Fig. \ref{fig:RE-K4-Peak} with $\Peak=4P/N$, it is observed that when $\overline{E}$ is sufficiently large, the TS scheme becomes worse than the PS scheme. This is because that for a finite peak power constraint on each SC, as $\overline{E}$ becomes sufficiently large, the TS scheme needs to schedule a nonzero EH slot to ensure all users harvest sufficient energy (see Fig. \ref{fig:alpha}), the total information transmission time $1-\alpha_{K+1}$ then decreases and results in a degradation of achievable rate.
However, for $80<\overline{E}\leq208\mu$W, in which case the system achieves large achievable rate (larger than $70\%$ of UB) while each user harvests a reasonable value of energy (about $32\%$ to $84\%$ of the maximum possible value), the TS scheme still outperforms the PS scheme.

\begin{figure}
\begin{center}
\scalebox{0.6}{\includegraphics{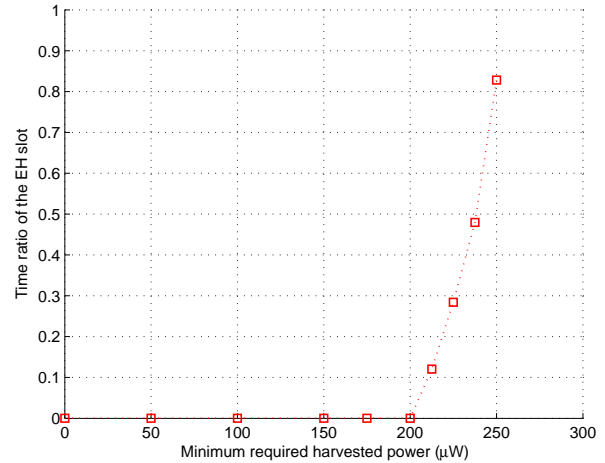}}
\end{center}
\caption{Time ratio of the EH slot versus minimum required harvested energy for the TS scheme in Fig. \ref{fig:RE-K4-Peak}.}
\label{fig:alpha}
\end{figure}

\begin{figure}
\begin{center}
\scalebox{0.6}{\includegraphics{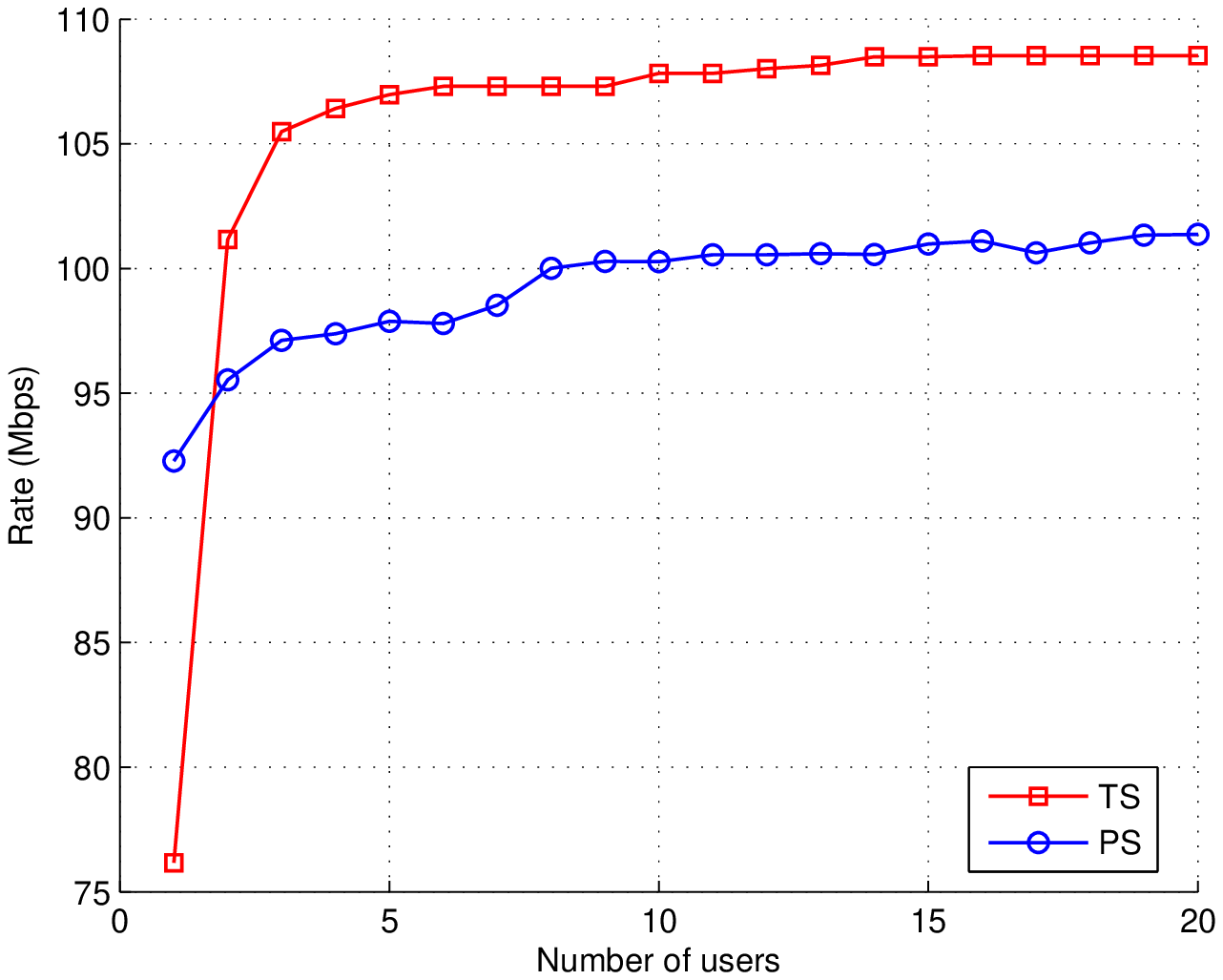}}
\end{center}
\caption{Achievable rates versus number of users, where $\Peak=4P/N$ and $\Emin=\overline{E}=150\mu$W.}
\label{fig:RateOverK}
\end{figure}

Fig. \ref{fig:RateOverK} shows the achievable rates versus the number of users by different schemes under fixed minimum required harvested energy $\Emin=\overline{E}=150\mu$W and $\Peak=4P/N$. In Fig. \ref{fig:RateOverK}, it is observed that for both TS and PS schemes, the achievable rate increases as the number of users increases,
and the rate tends to be saturated due to the bandwidth and the transmission power of the system is fixed.
In particular, for the TS scheme, the achievable rate with $K=2$ is much
larger (about $32.8\%$) than that with $K=1$. This is because that for the case $K=2$, one of the user decodes information when the other user is harvesting energy; however, for the single-user case $K=1$, the transmission time when the user is harvesting energy is not utilized for information transmission.
It is also observed in Fig. \ref{fig:RateOverK} that for a general multiuser system with large $K\geq 2$, the TS scheme outperforms the PS scheme.
This is intuitively due to the fact that as the number of users increases, the portion of energy discarded at the information receiver at each user after power splitting also becomes larger (c.f. Fig. \ref{fig:PA-PS}), hence using PS becomes inefficient for large $K$.


\section{Conclusion}\label{sec:conclusion}
This paper has studied the resource allocation optimization for a multiuser OFDM-based downlink SWIPT system. Two transmission schemes are investigated, namely, the TDMA-based information transmission with TS applied at each receiver, and the OFDMA-based information transmission with PS applied at each receiver. In both cases, the weighted sum-rate is maximized subject to a given set of harvested energy constraints as well as the peak and/or total transmission power constraint.
Our study suggests that, for the TS scheme, the system can achieve the same rate as the conventional TDMA system, and at the same time each user is still able to harvest a reasonable value of energy. When the harvested energy required at users is sufficiently large, however, a nonzero EH slot may be needed. This in turn degrades the rate of the TS scheme significantly. Hence, the PS scheme may outperform the TS scheme when the harvested energy is sufficiently large. From the view of implementation, the TS scheme is easier to implement at the receiver side by simply switching between the two operations of EH and ID. Moreover, in practical circuits the power splitter or switcher may introduce insertion loss and degrade the performance of the two schemes. This issue is unaddressed in this paper, and is left for future work.

\appendices
\section{Proof of Lemma \ref{lemma:1}}\label{appendix:proof 1}
To prove the concavity of function $f(q_{k,n},\alpha_k)$, it suffices to prove that for all $q_{k,n}\geq0$, $\alpha_k\geq0$, and the convex combination $(\hat{q}_{k,n},\hat{\alpha}_k)=\theta(\dot{q}_{k,n},\dot{\alpha}_k)+(1-\theta)(\ddot{q}_{k,n},\ddot{\alpha}_k)$ with
$\theta\in(0,1)$, we have $f(\hat{q}_{k,n},\hat{\alpha}_k)\geq \theta f(\dot{q}_{k,n},\dot{\alpha}_k)+(1-\theta)f(\ddot{q}_{k,n},\ddot{\alpha}_k)$.
With $q_{k,n}\geq0$, we consider the following four cases for $\alpha_k$.

  1) $\dot{\alpha}_k>0$ and $\ddot{\alpha}_k>0$: In this case, we have $\hat{\alpha}_k>0$. Since $\log_2\left(1+\frac{h_{k,n}q_{k,n}}{\N}\right)$ is a concave function of $q_{k,n}$, it follows that
      its perspective $\alpha_k\log_2\left(1+\frac{h_{k,n}q_{k,n}}{\N \alpha_k}\right)$ is jointly concave
      in $q_{k,n}$ and $\alpha_k$ for $\alpha_k>0$ \cite{Boydbook}. Therefore, we have $(\hat{q}_{k,n},\hat{\alpha}_k)\geq\theta(\dot{q}_{k,n},\dot{\alpha}_k)+(1-\theta)(\ddot{q}_{k,n},\ddot{\alpha}_k)$.

  2) $\dot{\alpha}_k>0$ and $\ddot{\alpha}_k=0$: In this case, we have $f(\ddot{q}_{k,n},\ddot{\alpha}_k)=0$, $\hat{\alpha}_k=\theta\dot{\alpha}_k$, and
  \begin{align*}\label{}
    f(\hat{q}_{k,n},\hat{\alpha}_k)
        &=\theta\dot{\alpha}_k\log_2\left(1+\frac{h_{k,n}(\theta\dot{q}_{k,n}+(1-\theta)\ddot{q}_{k,n})}{\N \theta\dot{\alpha}_k}\right)    \\
        &=\theta\dot{\alpha}_k\log_2\left(1+\frac{h_{k,n}\dot{q}_{k,n}}{\N \dot{\alpha}_{k,n}}+\frac{(1-\theta)h_{k,n}\ddot{q}_{k,n}}{\N \theta\dot{\alpha}_k}\right)
  \end{align*}
  Thus, we have $f(\hat{q}_{k,n},\hat{\alpha}_k)\geq \theta f(\dot{q}_{k,n},\dot{\alpha}_k)+(1-\theta)f(\ddot{q}_{k,n},\ddot{\alpha}_k)$.

  3) $\dot{\alpha}_k=0$ and $\ddot{\alpha}_k>0$: Similar as case 2), we have $f(\hat{q}_{k,n},\hat{\alpha}_k)\geq \theta f(\dot{q}_{k,n},\dot{\alpha}_k)+(1-\theta)f(\ddot{q}_{k,n},\ddot{\alpha}_k)$.

  4) $\dot{\alpha}_k=0$ and $\ddot{\alpha}_k=0$: In this case, we have $f(\hat{q}_{k,n},\hat{\alpha}_k)= f(\dot{q}_{k,n},\dot{\alpha}_k)=f(\ddot{q}_{k,n},\ddot{\alpha}_k)=0$. Therefore, $f(\hat{q}_{k,n},\hat{\alpha}_k)=\theta f(\dot{q}_{k,n},\dot{\alpha}_k)+(1-\theta)f(\ddot{q}_{k,n},\ddot{\alpha}_k)$.

From the above four cases, we have $f(\hat{q}_{k,n},\hat{\alpha}_k)\geq \theta f(\dot{q}_{k,n},\dot{\alpha}_k)+(1-\theta)f(\ddot{q}_{k,n},\ddot{\alpha}_k)$ for all
$q_{k,n}\geq0$ and $\alpha_k\geq0$, and thus $f(q_{k,n},\alpha_k)$ is concave, which
completes the proof.

\section{Proof of Proposition \ref{proposition-subgradient}}\label{appendix:proof 2}
For any $\hat{\lambda}_i\geq0,i=1,\ldots,K,\hat{\mu}\geq0,\hat{\nu}\geq0$, we have

\begin{small}\vspace{-0.1in}
\begin{align}
  g\left(\hat{\lambda}_i,\hat{\mu},\hat{\nu}\right)
  & \geq \mathcal{L}\left(\{\dot{q}_{k,n}\},\{\dot{\alpha}_k\},\{\hat{\lambda}_i\},\hat{\mu},\hat{\nu}\right) \nonumber\\
  & = g\left(\{\lambda_i\},\mu,\nu\right)    \nonumber\\
  & ~~ ~~ +\sum\limits_{i=1}^{K}(\hat{\lambda}_i-\lambda_i)\left(\zeta\sum\limits_{k\neq i}^{K+1}\sum\limits_{n=1}^{N}h_{i,n}\dot{q}_{k,n}-\overline{E}_i\right) \nonumber\\
  & ~~ ~~ +(\hat{\mu}-\mu)\left(P-\sum\limits_{k=1}^{K+1}\sum\limits_{n=1}^{N}\dot{q}_{k,n}\right) \nonumber\\
  & ~~ ~~ +(\hat{\nu}-\nu)\left(1-\sum\limits_{k=1}^{K+1}\dot{\alpha}_k\right)    \nonumber
\end{align}
\end{small}%
By the definition of subgradient, the choice of $\mv{d}$ as given in (\ref{eq:subgradient-TS}) is indeed a subgradient for $g\left(\{\lambda_i\},\mu,\nu\right)$.

\end{document}